 \documentclass[final]{l4dc2026}

\usepackage[english]{babel}
\usepackage{wrapfig, ragged2e}
\usepackage{lmodern}
\usepackage[scaled]{}
\usepackage[utf8]{inputenc}

\usepackage{soul}
\usepackage{siunitx}

\usepackage{soul}
\usepackage{siunitx}

\usepackage{verbatim}
\usepackage[mathscr]{euscript}
 \let\mathscr\relax
\usepackage[scr]{rsfso}
\usepackage{enumitem}

\newcommand\numberthis{\addtocounter{equation}{1}\tag{\theequation}}

\newtheorem{problem}{Problem}

\usepackage[many]{tcolorbox}   

\usepackage{multicol}               

\definecolor{main}{HTML}{000000}    
\definecolor{sub}{HTML}{e1e9f1}     

\tcbset{
    sharp corners,
    colback = white,
    before skip = 0.2cm,    
    after skip = 0.5cm      
}                           

\newtcolorbox{boxL}{
    fontupper = \color{main},
    rounded corners,
    arc = 6pt,
    colback = sub, 
    colframe = main!50, 
    boxrule = 0pt, 
    bottomrule = 4.5pt 
}

\usepackage[mathscr]{euscript}
 \let\mathscr\relax
\usepackage[scr]{rsfso}

\usepackage[utf8]{inputenc}
\usepackage{pgfplots}

 \usepackage{booktabs} 
 \usepackage{multirow}
\usepackage{float}
\usepackage{subcaption}

\newtheorem{thm}{Theorem}

\newcommand{\RM}[1]{\textcolor{red}{\textbf{RM:} #1}}

\def\tuple#1{{\langle #1 \rangle}}

\title[Kernel-Based Safe Exploration in Deep Reinforcement Learning]{Kernel-Based Safe Exploration in Deep Reinforcement Learning}
\usepackage{times}

\coltauthor{\Name{Rupak Majumdar} \Email{rupak@mpi-sws.org}\\
 \Name{Nikhil Singh} \Email{niksingh@mpi-sws.org}\\
 \Name{Sadegh Soudjani} \Email{sadegh@mpi-sws.org}\\
 \addr Max Planck Institute for Software Systems, Germany}

\begin{document}

\maketitle

\begin{abstract}
Safety has been a major concern when deploying deep reinforcement learning algorithms in the real world. 
A promising direction that ensures that the learned policy does not visit unsafe regions is to learn a \emph{barrier function} along with the
policy.
A barrier is a function from states to reals that assigns low values
to the initial states, high values to the unsafe states, and decreases in expectation on each transition; such a function can be used to bound the probability of reaching unsafe states.
Previous attempts learned a barrier function directly from exploration data,
but this required either large amounts of data or restrictions on the system dynamics.
In this paper, we show how kernel embeddings can be used to learn
barrier functions during deep reinforcement learning for stochastic systems with unknown dynamics.
Our algorithm, \emph{kernel-based safe exploration (KBSE)}, learns an optimal policy and a barrier simultaneously during exploration.
The barriers are computed iteratively, represented as conditional mean embeddings, and
provide better probabilistic safety guarantees with more exploration.
The exploration algorithm uses the learned barrier functions to identify safety violations. In the case of violation, it intervenes to modify the unsafe action to a safe action, thereby ensuring that the exploration is restricted to actions that bound the probability of reaching unsafe states. 
We evaluate KBSE on several complex continuous control benchmarks. 
Experimental results establish our new algorithm to be suitable for synthesizing control policies that are probabilistically safe without degradation in reward accumulation.
\end{abstract}

\section{Introduction}
\label{sec:intro}

In recent years, many challenging problems in optimal policy synthesis,
including solving Atari~\citep{Mnih15},
Go~\citep{Silver17}, biped walking~\citep{Levine16} and Language Models~\citep{naveed24},
have been solved using reinforcement learning (RL).
This is attributed to RL's ability to maximize cumulative rewards through online exploration, 
without requiring a model of the underlying dynamics.

For safety-critical, high-dimensional, control problems, an obstacle to deploying RL-based controllers is the possibility
that a reward-maximizing controller may lead the system into unsafe states. 
Thus, a challenging research direction is to learn an optimal policy that is constrained by safety requirements.
While there already exists a rich literature on this problem, existing solutions are still unsatisfactory:
either because they do not provide certificates of correctness, or because they require data-intensive techniques, or
because the learning process is subject to oscillations or high approximation errors.

In this paper, we propose an online, kernel-based approach to learning safe control policies under unknown dynamics.
Our algorithm, \emph{kernel-based safe exploration (KBSE)}, learns an optimal policy and a barrier function simultaneously.
The policy optimizes reward accumulation while being constrained by a barrier function.
The barrier function \citep{Prajna} provides a certificate that the system violates the safety specification with a bounded probability with high confidence \citep{schon24}.
In particular, the exploration is constrained to use actions that do not violate the safety specification.
Simultaneously, the exploration is used to update the barrier function to provide better safety guarantees over time.

Intuitively, a barrier function maps states to reals such that the function (a) assigns a low value to initial states, 
(b) assigns a high value to unsafe states,
and (c) reduces in expectation on each step of the dynamics---the existence of such a function implies an upper bound on the probability that
an unsafe state may be reached.
The technical challenge in learning barrier functions is the conditional expectation in requirement (c).
The key to our technique is the representation of barrier functions 
using conditional mean embeddings (CMEs) \citep{Song2009CondEmbed,Klebanov2020RigorousCME} in reproducing kernel Hilbert spaces (RKHS) \citep{scholkopf2002learning,Berlinet2004RKHSProbStat,Steinwart2008SVM}.
An RKHS is a Hilbert space of functions characterized by the property that pointwise evaluation is a continuous linear functional. 
This property renders RKHS a powerful framework for studying and manipulating functions in high-dimensional spaces. 
A CME embeds conditional probability distributions into an RKHS, allowing for the computation of conditional expectations via simple inner products. 
In particular, using CMEs, we can synthesize barrier functions using a simple linear optimization problem.

We evaluate KBSE on several challenging continuous control benchmarks available in the Gym~\citep{gym} classical control and Mujoco-based environments. 
For each of these benchmarks, we consider safety specifications with requirements more stringent than those used in Gym environments. 
Our approach generalizes the safety requirements by allowing the violation of safety constraints with a certain probability threshold. This enables us to study systems for which no policy exists to satisfy safety almost surely.
We compare the performance of the control policy synthesized by KBSE against the most popular off-policy safe RL algorithms.
Experimental results show that the KBSE algorithm is superior to the baseline algorithms in terms of reward accumulation and safety cost. In addition, KBSE also provides the safety probability for the learned control policies which is not possible in the case of baseline algorithms. We have relegated the proofs of statements and additional details to the supplementary material.





\section{Related Work}

\noindent
\textbf{Constrained RL.}
Learning policies subject to constraints has been studied widely using a range of techniques including constrained policy optimization~\citep{Achiam17}, Lagrangian multipliers~\citep{Ray2019,stooke20}, and penalty-based methods \citep{Liu19}. While these approaches are subject to inherent oscillations, initial value dependency, large approximation errors, or large computational costs, they are not able to encode safety probability thresholds in their policy learning. Our approach integrates directly the bounds on safety probabilities in the learning through the concept of barrier functions.   

%

\smallskip
\noindent
\textbf{Barrier functions in RL.} 
Recent strategies in safe RL use various types of safety certificates to  ensure constraint satisfaction during exploration.
The work by \citet{li19,Yang23b} uses barrier and Lyapunov functions combined with RL.
\citet{Cheng19} propose RL-CBF, providing safety guarantees with high probability assuming deterministic dynamics (known nominal dynamics and partially unknown dynamics).
The approach by \citet{Wang23} uses barrier functions for safe RL in stochastic dynamics using generative models.
These approaches are subject to handling deterministic systems, being limited to low-dimensional systems, or requiring large amounts of data to learn parameters of deep neural networks (DNNs) representing the barrier certificates and shields.
In contrast, our approach studies unknown stochastic systems using non-parametric kernel-based methods \citep{hofmann08,Meanti20}, which makes it applicable to large-dimensional systems.




\smallskip
\noindent
\textbf{Kernel methods in RL.} 
Kernel methods have been explored for policy synthesis \citep{Ormoneit02,Barreto16}.
The work by \citet{Romao23} uses kernel mean embedding in value iteration to find the optimal policy, thus being applicable to small-dimensional benchmarks.
The work by \citet{Domingues21} provides a model-based algorithm and uses kernel smoothing to estimate rewards and transitions and obtains a regret bound for kernel-based RL.
\citet{Thorpe21,Thorpe23} use conditional distribution embeddings to
represent a model-free controller synthesis problem as a linear program in RKHS. 
However, this approach does not take into account safety constraints, which we address here.

\smallskip
\noindent
\textbf{Novelty of our formulation.}
Our problem formulation is distinct from the standard safe RL approaches in the literature. In safe RL with shielding~\citep{Alshiekh18,reed24}, the safety requirement is specified as the level set of a function, and it disables actions that lead the system to an unsafe state. This means that the designed policy must satisfy the safety constraint almost surely (with probability 1). In contrast, our approach generalizes the safety requirements by allowing the violation of safety constraints with a certain probability threshold. 
In our problem formulation, we  intentionally permit a certain level of safety violation due to the constraints in the definition of the barrier function. Our aim is to learn a control policy that gives safety with a certain probability without degradation in performance.

Our approach is also distinct from the available constrained RL techniques \citep{Achiam17,Liu19,stooke20} as they consider constraints in the form of a bound on an additive or average objective. The probabilistic safety constraint in our formulation cannot be written in the form of additive or average objectives. Thus, these approaches are not applicable to our problem formulation.

\section{Preliminaries}

\textbf{Notation.} 
The set of reals, non-negative reals, and non-negative integers are denoted respectively by $\mathbb R$, $\mathbb R_{\ge 0}$ and $\mathbb N$.
Consider a Polish sample space $\mathbb{X}$ with underlying probability space $(\mathbb{X},\mathcal{B}(\mathbb{X}),\mathbb{P})$ endowed with a Borel $\sigma$-algebra $\mathcal{B}(\mathbb{X})$ and a probability measure $\mathbb{P}$.
We denote a random variable by $X$ and the instantiations are denoted by $x$. For a random variable $X$, let $p_{X}$ be pushforward probability measure of $\mathbb{P}$ under $X$ such that $X \sim p_{X}(\cdot)$. For a function $f(X)$, the expected value on $\mathbb{X}$ is given by $\mathbb{E}_{p_{X}}[f(X)]$.
For a measurable space $(\mathbb{X},\mathcal{B}(\mathbb{X}))$, the set of all probability measures is denoted by $\mathcal{P}(\mathbb{X})$.
For two measurable spaces $(\mathbb{X},\mathcal{B}(\mathbb{X}))$, $(\mathbb{Y},\mathcal{B}(\mathbb{Y}))$, a probability kernel is defined as the mapping $p : \mathbb{X} \times \mathcal{B}(\mathbb{Y}) \rightarrow [0,1]$  
such that $p(x,\cdot) : \mathcal{B}(\mathbb{Y}) \rightarrow [0,1]$ is a probability measure for all $x \in \mathbb{X}$ and $p(\cdot,B) : \mathbb{X} \rightarrow [0,1]$ 
is measurable for all $B \in \mathcal{B}(\mathbb{Y})$.
For each $x\in \mathbb{X}$, the conditional probability measure $p(x,\cdot)$ is also denoted by $p(\cdot | x)$.




\smallskip \noindent \textbf{Constrained Reinforcement Learning.}
%
We represent the environment $\mathcal{M}$ as a Markov Decision Process (MDP) $\mathcal{M} = \langle S,A,\mathcal{T},R\rangle$,
where $S \subset \mathbb{R}^p$ denotes the set of continuous states of the system and $A \subset \mathbb{R}^q$ denotes the set of continuous control actions. 
The function $\mathcal{T}: S\times A \times \mathcal{B}(S)\rightarrow[0,1]$ is a probability kernel, where $\mathcal{T}(s,a,\cdot)$ is the probability measure for transitioning to the next state from the current state $s\in S$ under action $a\in A$. The probability kernel $\mathcal T$ is assumed to be unknown.  
The function $R : S \rightarrow \mathbb{R}$ represents the reward function for $\mathcal{M}$.  
A transition (or sample) at time $t\in\mathbb N$ is denoted by $\langle s_t,a_t,r_t,s_{t+1} \rangle$, where $s_t\in S$ is a state, $a_t\in A$ is an action, $r_t \in \mathbb{R}$, such that $ r_t = R(s_t)$, is the reward obtained by the agent by performing action $a_t$ at state $s_t$, and $s_{t+1} \in S$ is the next state selected randomly according to the probability measure $\mathcal{T}(\cdot| s_t,a_t)$.

A Constrained Markov Decision Process (CMDP) is a tuple $\langle S, A, \mathcal{T}, R, \varphi\rangle$, which
is an MDP $\langle S, A, \mathcal{T}, R\rangle$ augmented with a \emph{safety specification} $\varphi = \tuple{S_u, T}$, in which $S_u\subset S$ is the set of unsafe states and $T\in\mathbb N$ is a time horizon.

A parameterized \emph{policy} is a probability kernel $\pi:S\times\Theta\times \mathcal B(A)\rightarrow[0,1]$, where $\Theta$ is the set of parameters, and for each $\theta\in\Theta$, $\pi_{\theta}(\cdot|s) = \pi(s,\theta,\cdot)$ is a conditional probability measure for selecting the action $a$ at state $s$,
given parameters $\theta$.
A probability measure $\rho_0:\mathcal B(S)\rightarrow[0,1]$ for the initial state $s_0$ and a policy $\pi_{\theta}$ together determines a probability measure $\chi$ over
trajectories $\omega=\langle s_0,a_0,s_1,\ldots,s_T \rangle$, given by
$\chi(\rho_0, \theta)(\omega) = \rho_0(ds_0)\prod_{t=0}^{T} \pi_{\theta}(da_t|s_t)\cdot \mathcal{T}(ds_{t+1}|s_t,a_t)$.

We denote the probability of event with respect to this measure by $\mathbb P_{\rho_0}^\theta[\cdot]$ and expectations of random variables by $\mathbb E_{\rho_0}^\theta[\cdot]$. When the initial state is fixed to a single state $s_0$, we replace $\rho_0$ with $s_0$. 

The safe reinforcement learning problem asks to compute an
optimal policy $\pi_{\theta^\ast}$ that maximizes the expected
discounted sum of rewards up to horizon $T$
while ensuring that the probability that the MDP reaches an unsafe state
within the horizon is at most a given bound $\delta$.
Formally,


%

\begin{boxL}
\begin{problem}[Constrained Policy Synthesis] 
\label{prob-def}
Given a CMDP $\tuple{S, A, \mathcal{T}, R, \tuple{S_u, T}}$ with unknown $\mathcal{T}$, initial state $s_0\in S$, parametrized policy $\pi_\theta$, discount factor $\gamma \in [0, 1)$, and a parameter $\delta$,
synthesize an optimal policy $\pi_{\theta^*}$ that maximizes the 
expected discounted sum of rewards while ensuring that the probability
of reaching $S_u$ is bounded above by $\delta$.
Mathematically, find an optimal policy $\pi_{\theta^*}$ such that
\begin{align*}
    \theta^* = \underset{{\theta\in\Theta}}{\arg\max} \  
    {\mathbb{E}}_{s_0}^\theta \ \left[\sum_{t=0}^{T} \gamma^t R(s_t)\right]
    \quad\text{s.t.}~~~
    \mathbb{P}_{s_0}^\theta
    \{\mbox{some state in $\omega$ is in } S_u\}\leq \delta. 
    \numberthis
    \label{prob}
\end{align*}
\end{problem}
\end{boxL}

The use of probability threshold enables us to generalize the safety requirements by allowing the violation of safety constraints up to a threshold and improve the optimal objective. This also enables us to study systems for which no policy exists to satisfy safety almost surely.
Solving Problem~\ref{prob-def} requires handling two challenges: (a) ensuring the probabilistic safety constraint, and (b) finding the policy under the unknown probability kernel $\mathcal T$. 
We utilize the concepts of barrier functions and Conditional Mean Embedding (CME) to tackle these challenges, as described below after giving a motivating example. 

\begin{wrapfigure}{r}{0.24\textwidth}
    \centering
    \includegraphics[height=1.6cm]{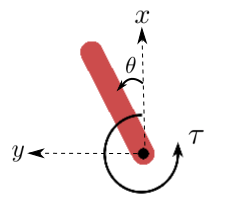}
    \caption{Pendulum}
    \label{fig:pend}
\end{wrapfigure}

\smallskip
\noindent
\textbf{A motivating example.}
\label{sec-pend}
We motivate our technique using the classical example of controlling an inverted pendulum~\citep{gym} to maintain its upright position.

A pendulum freely hangs in the downward position, and the default goal is to balance it vertically upward.
Normally, the control policy learns to balance upright by swinging the pendulum by one full round to gain sufficient momentum to bring it to the vertical position. However, if the starting configuration of the pendulum is in upward direction (above horizontal), a full swing might not be needed, and it can balance it upright with less effort. 
In this case, our safety specification requires that the pendulum should not go for a full swing.
Given that the starting state of the pendulum is feasible (sufficiently above horizontal position), our goal is to learn a control policy that balances the pendulum in the vertical upward position without going for a full swing.




\smallskip
\noindent
\textbf{Barrier functions.}
%
The notion of 
Barrier function provides a technique to define locally checkable conditions that guarantee satisfaction of the safety requirements in CMDPs. 
Fix a CMDP $\tuple{S, A, \mathcal{T}, R, \tuple{S_u, T}}$,
a policy $\pi_\theta$, and an initial
state $s_0$.
A function $B : S \rightarrow \mathbb{R}_{\geq 0}$ is a 
\emph{barrier function} w.r.t. the unsafe set $S_u$ if it satisfies the following conditions for some constants $\nu>\eta\geq 0$ and $c\geq 0$:

    \textbf{(i)} $B(s_0)\leq \eta$;~~
    \textbf{(ii)} $\forall s\in S_u:\,\,  B(s)\geq \nu$;~~ and
    \textbf{(iii)} $\forall s \in S: ~~\mathbb{E}_{s}^{\theta} [B(s^+) | s] - B(s)\leq c $;
    
where $s^+$ is the next state of the CMDP from the current state $s$ for the policy $\pi_\theta$.
\begin{thm}[\cite{Kushner}]
\label{thm:safe-prob}
If there is a non-negative barrier function $B : S \rightarrow \mathbb{R}_{\geq 0}$ for the CMDP $\tuple{S, A, \mathcal{T}, R, \tuple{S_u, T}}$, then
$
\mathbb{P}_{s_0}^{\theta}[
\mbox{some state in $\omega$ is in $S_u$}] \leq \frac{\eta + cT}{\nu}.
$
\end{thm}

Theorem~\ref{thm:safe-prob} gives a bound on safety using barrier certificates and assuming full knowledge of the model.
Problem~\ref{prob-def} assumes the model of the system is unknown and requires solving a safety-constrained optimization.
Our proposed approach fills the gap between Theorem~\ref{thm:safe-prob} and the solution to Problem~\ref{prob-def}.


\smallskip
\noindent
\textbf{Reproducing Kernel Hilbert Spaces (RKHS).}
Let $\mathcal{H}$ be a Hilbert space of functions $f:S\rightarrow \mathbb{R}$
with inner product $\langle\cdot,\cdot\rangle_{\mathcal{H}}$ and a kernel $k : S \times S \rightarrow \mathbb{R}$. 
Each RKHS is equipped with a dot product which satisfies the reproducing property:
\begin{align*}
    \langle f(\cdot), k(s,\cdot) \rangle_{\mathcal{H}} = f(s) \quad\text{ and }\quad
    \langle k(s,\cdot),k(s',\cdot) \rangle_{\mathcal{H}} = k(s,s')~.
\end{align*}

Intuitively, a function evaluation at any point $s\in S$ can be expressed as an inner product. The function $k$ is called the \emph{reproducing kernel} since it reproduces the value of $f$ at $s$ via the inner product.
For any set $Y$, we will use the notations $k_Y$ and $\mathcal H_{k_Y}$ to denote respectively the kernel and its RKHS on $Y$.

\begin{definition}[Mean embedding (ME)]
Consider a kernel $k_{S} : S \times S \rightarrow \mathbb{R}$ and the RKHS $\mathcal{H}_{k_S}$ on $S$. The mean embedding (ME) of a probability measure $p : \mathcal{B}(S)\rightarrow[0,1]$ is given via the mean map $\mu_{k_{S}}:\mathcal P(S)\rightarrow\mathcal H_{k_{S}}$:
\begin{equation}
\mu_{k_{S}}(p) = \underset{s\sim p}{\mathbb{E}} [k_{S}(\cdot,s)] = \int_{S} k_{S}(\cdot,s)\cdot dp(s)~.     
\end{equation}
\end{definition}
The reproducing property carries on to the ME,
allowing to compute the expected value of a function $f\in\mathcal H_{k_S}$ via its inner product with the ME :~ 
    $\mathbb E_{s\sim p}[f(s)] = \langle f, \mu_{k_{S}}(p)\rangle_{\mathcal H_{k_S}}$.
Using ME, we define the maximum mean discrepancy (MMD) to be the distance between two probability measures $p,p'$ in $\mathcal{H}_{k_S}$ as $\parallel \mu_{p}(\cdot) - \mu_{p'}(\cdot) \parallel_{\mathcal{H}_{k_S}}$~\citep{gretton12}.

\begin{definition}[Conditional mean embedding (CME)]
Given the RKHS $\mathcal{H}_{k_{SA}}$ on the space $S\times A$ with a kernel $k_{SA}$, and the RKHS $\mathcal{H}_{k_S}$ on the space $S$ with a kernel $k_S$,
The CME of a conditional probability measure $p:S\times A\times \mathcal{B}(S)\rightarrow[0,1]$ is an $S\times A$-measurable random variable taking values in $\mathcal{H}_{k_S}$ given by
    \begin{equation}
    \label{cme1}
        \mu(p)(\cdot,s,a) = \underset{s^+\sim p(s,a,\cdot)}{\mathbb{E}} [k_{S}(\cdot,s^+)| s,a]~. 
    \end{equation}
\end{definition}

Analogous to the non-conditional case,
we can compute the expected value of a function $f \in \mathcal{H}_{k_S}$
via its inner product with the CME:~
    $\underset{s^+\sim p(s,a,\cdot)}{
    \mathbb{E}} [f(s^+)|s,a] = \langle f(\cdot), \mu(p)(\cdot, s,a) \rangle_{\mathcal{H}_{k_S}}$.
To establish barrier functions, we start by reformulating
the left-hand side of condition 3 of the barrier function definition as an inner
product with the CME, i.e., we have
\begin{equation}
\label{eq:expected_barrier}
\mathbb{E}_{\mathcal{T}} [
 B(s^{+}) | s,a] = \langle B,  \mu(s,a)  \rangle_{\mathcal{H}_{k_{S}}}~. 
\end{equation}


\begin{remark}
\label{rem:G}
For mathematical consistency, we restrict the unknown model of the system to a known function space. In particular, we assume that the CME of the probability kernel $\mathcal{T}$ lives in a vector-valued RKHS of functions from $S \times A$ to $ \mathcal{H}_{k_{s}}$, denoted by $\mathcal{G}$ \citep{Park20}.
\end{remark}

We utilize the above observation and empirical computation of the CME to provide a solution for Problem~\ref{prob-def}, as described in the next section.

\section{Kernel-based Safe Exploration}


Our goal is to solve Problem~\ref{prob-def} by simultaneously learning an
optimal policy and a barrier function using data from the unknown CMDP.
However, learning the barrier function is a data-intensive process~\citep{kordabad2024control,salamati2024data,Dai23} and estimating the conditional
expectation in the condition (iii) of the barrier function is difficult.
To tackle these challenges, we use RKHS and CME to learn a valid barrier function. We show in this section that the CME transforms the problem of finding a valid barrier function into a linear program which can be solved efficiently. 
We also show that the barrier function learned via the approximated CME converges to the true barrier function as the size of the dataset increases (cf. Theorem~\ref{thm-conv}). 
%




The overall algorithm, called \emph{kernel-based Safe Exploration} (\textit{KBSE}), learns a controller online in a Deep-RL setting for the CMDP $\mathcal M$ with an unknown probability kernel $\mathcal T$. In the following subsections, we describe the main steps of the KBSE algorithm in detail, with the algorithm included in Appendix~\ref{sec-KBSE-algorithm}.


\subsection{Data collection}
\label{sec-datacoll}
We use data generated during exploration, denoted as $\mathcal R = \{\langle \hat s_i,\hat a_i,\hat r_i,\hat s_{i}^+ \rangle,i=1,2,\ldots,N\}$. We can also denote the dataset alternatively as  $ \mathcal{R} = (\hat{S},\hat{A},\hat R, \hat{S}^+)$, with
\begin{equation}
\label{eq-data-coll}
    \hat{S} = [ \hat{s}_1,\ldots,\hat{s}_N ]^T~,
    ~\hat{A} = [ \hat{a}_1,\ldots,\hat{a}_N ]^T~,
    ~\hat{R} = [ \hat{r}_1,\ldots,\hat{r}_N ]^T~,
    ~\hat{S^+} = [\hat{s}_1^+,\ldots,\hat{s}_N^+ ]^T~,
\end{equation}
where $\hat r_i = R(\hat s_i)$ and $s_i^+\sim\mathcal T(\hat s_i,\hat a_i,\cdot)$ for all $i$.
As an off-policy RL method that stores all the transitions collected during training in a replay buffer, we randomly sample from this buffer to generate independent and identically distributed (iid) samples.

For the CME and approximation of the expectation, we select the kernel $k_S$ to be a radial basis function (RBF) of the form
    $k_{S}(s_i,s_j) = \exp \left[ \frac{-\| s_i - s_j \|^2}{2\cdot \sigma^2} \right],$ for all $s_i,s_j \in S$,
and similarly for the kernel $k_{SA}$ of the product space $S\times A$.
We define the corresponding \emph{Gram matrices}
$K_{\hat{S}}$ and $K_+$ as
    $K_{\hat{S}} := [k_{S}(\hat s_i,\hat s_j)]_{ij} \text{ and }
    K_{+} := [k_{S}(\hat s_i^+,\hat s_j^+)]_{ij}$.
We also define the vector-valued functions     $k_{\hat{S}}(s) := [k_{S}(s,\hat s_i)]_{i}$  and 
    $k_{+}(s) := [k_{S}(s,\hat s_i^+)]_{i}$.
For the product space $S\times A$, the Gram matrix $K_{\hat S\hat A}$ and the function $k_{\hat S\hat A}(s,a)$ are defined similarly.

The function $\mathtt{sample\_data }$ used in Algorithm~\ref{algo1} in Appendix~\ref{sec-KBSE-algorithm} implements the generation of iid samples from $\mathcal{R}$.

\subsection{Generating a Valid Barrier Function}
\label{sec-learn-barr}

To get a valid barrier function, we learn a barrier function from the data and perform a validity check iteratively until a valid barrier function is generated.

 \smallskip
 \noindent
\textbf{Learning barrier function.}
We define the barrier $B$ as a linear combination of the RBF kernel functions, given as 
\begin{equation}
\label{eq:barrier_linear}
    B(s) = \sum_{i=1}^{N} \alpha_i \cdot k_s(s,s_i)~.
\end{equation}
Using the safety specification $\varphi = \tuple{S_u, T}$, we classify each state $\hat s_i \in S$ using a binary safe/unsafe label $Y$.
We use linear regression with regularization~\citep{montgomery12} to solve for $\alpha$ as
\begin{equation}
    \underset{\alpha}{\min} \| \alpha \cdot K_{\hat{S}}  - Y \|^2 + \lambda \| \alpha \|^2~.
\end{equation}
We compute $\eta = B(s_0)$ and $\nu = \min_{s\in S_u} B(s)$. The function in \eqref{eq:barrier_linear} satisfies the conditions (i) and (ii) of being a barrier function if $\nu>\eta$.

\smallskip
\noindent
\textbf{Barrier function validation.}
%
In order to compute the constant $c$ in the condition (iii) of the barrier function and validate the function in \eqref{eq:barrier_linear},
we compute a data-driven empirical CME $\hat \mu(\cdot,s,a)$ for the CME $\mu$ in \eqref{cme1} applied to the probability kernel $\mathcal T$.
%
%
%
%
%
Using data $\mathcal{R} = (\hat{S},\hat{A},\hat R, \hat{S}^+)$, the empirical CME~\citep{Grunewalder12,Thorpe21} is
\begin{equation}
\label{eq:empirical_CME}
    \hat{\mu}(\cdot,s,a) = k_{\hat{S}\hat{A}}(s,a)^{T} [K_{\hat{S}\hat{A}}+\lambda N \mathbb{I}_N]^{-1}  k_{+}(\cdot)~,
\end{equation}
where $\lambda \ge 0$ is a regularization constant,
$\mathbb{I}_N$ is the identity matrix with dimension $N$.
%
%
%
Using the reproducing property, we have the following approximation for the conditional expectation in
\eqref{eq:expected_barrier}:
\begin{equation}
\label{eq-inner-prod}
 \langle B, \hat{\mu}(\cdot, s,a) \rangle_{\mathcal{H}_{k_{S}}}
 = k_{\hat{S}\hat{A}}(s,a)^{T} [K_{\hat{S}\hat{A}} +\lambda N \mathbb{I}_N]^{-1} B(\hat{S}^+)~.
\end{equation}

For simplicity, we denote
\begin{equation}
\label{eq:W}
    W(s,a)^T=k_{\hat{S}\hat{A}}(s,a)^{T} [K_{\hat{S}\hat{A}} +\lambda N \mathbb{I}_N]^{-1}~.
\end{equation} 


The empirical CME $\hat{\mu}$ deviates from the true CME $\mu$ by at most $\epsilon$ in the $\mathcal{G}$-norm with probability $1-\zeta$, where the approximation precision $\epsilon$ is an error bound to represent the Maximum Mean Discrepancy (MMD)~\citep{Nemmour22} radius of the RKHS ambiguity set centered at the empirical CME.
We therefore need the barrier condition to hold robustly over all CMEs within this $\epsilon$-ball.
We define an ambiguity set $\mathcal{C}_{\epsilon}$, with MMD $\epsilon$, centered at the empirical CME $\hat{\mu}(\cdot, s,a)$ such that the true CME $\mu(\cdot, s,a)$ lies within the ambiguity set with probability at least $(1-\zeta)$, i.e., 
\begin{equation}
    \mathbb{P}(\mu(s,a) \in \mathcal{C}_{\epsilon}) \geq 1- \zeta, \text{ where } \mathcal{C}_{\epsilon} = \bigl\{ \mu \in \mathcal{G} | \parallel \mu - \hat{\mu} \parallel_{\mathcal{G}} \leq \epsilon
  \bigl\}. 
  \label{eq:amb}
\end{equation}

\begin{thm}
    Consider the dataset $ \mathcal{R} = (\hat{S},\hat{A},\hat R, \hat{S}^+)$ with the kernels  $k_S$ and $k_A$ and the Gram matrices $K_{\hat{S}}$, $K_{\hat{A}}$ and $K_+$. 
    Consider the ambiguity set $\mathcal{C}_{\epsilon}$ centered at the empirical CME $\hat\mu$ in \eqref{eq:empirical_CME} with confidence $(1-\zeta)$. If there exists a function $B : S \rightarrow \mathbb{R}_{\geq 0}$, $B \in \mathcal{H}_{k_S}$ with $\bar{B} \geq \| B \|_{\mathcal{H}_{k_S}}$, such that for all $s\in S$ there is an $a\in A$ satisfying
    \begin{equation}
    \label{bar_eq}
     W(s,a)^T\cdot B(\hat{S}^+) - B(s) \leq c - \epsilon\cdot \sqrt{k_{SA} ((s,a),(s,a))}\cdot \bar{B},
    \end{equation}
    for some $c\geq 0$ with $W(s,a)$ in \eqref{eq:W},
    then B satisfies the condition (iii) of the barrier function with constant $c$ and confidence $(1-\zeta)$.
\end{thm}

\begin{proof}
    Refer Appendix~\ref{sec-proofs}.
\end{proof}



To compute the constant $c$ for the barrier function $B$, we convert the constraint in \eqref{bar_eq} into a min-max optimization problem, as follows:
\begin{equation}
 c = \underset{s\in S}{\min}~~\underset{a\in A}{\max}\left[ W(s,a)^T\cdot B(\hat{S}^+)
 - B(s) + \epsilon\cdot \sqrt{k_{SA} ((s,a),(s,a))}\cdot \bar{B}\right]. 
\label{final-eq}
\end{equation}

The function $\mathtt{compute\_BC} $ in Algorithm~\ref{algo1} in Appendix~\ref{sec-KBSE-algorithm} generates a valid barrier function.

\subsection{Finding the Safe Actions}

Once we compute a valid barrier function, we use the inequalities of the barrier function to identify and reduce the number of safety violations. 
To handle these safety violations and guide locally the reinforcement learning towards safe states, we need the local dynamics of the system.
We learn a local linear dynamics of the CMDP, represented by matrices $P$ and $Q$, using state transition data $\langle \hat{s}_t,\hat{a}_t,\hat{s}_{t+1} \rangle$. We use regression and solve the following optimization problem~\citep{montgomery12}:
\begin{equation}
 \label{eq:regr}
    \underset{P,Q}{\min}\sum_{i=t-H}^{t-1} \| \hat{s}_{i+1} -P\cdot \hat{s}_i - Q\cdot \hat{a}_i \|_{2}^2~,
\end{equation}
where $H$ denotes the number of local transitions used to learn the matrices $P$ and $Q$.


Assuming safety specification violation at the state $s_t$, we solve a quadratic optimization problem that modifies the current action $a_t$  to find the action $\bar{a}$ closest to $a_t$ leading to a safe state $s_{t+1}$, given as
\begin{equation}
    \label{eq:qpprob}
    \bar{a}_t = \underset{\bar{a}}{\arg\min} \| \bar{a} - a_t \|^2\nonumber
    ~~~~~~~\text{s.t.}~~~ B(s_{t+1}) \leq \nu, 
    ~~~ s_{t+1} = P\cdot s_t + Q\cdot \bar{a}~.\numberthis
\end{equation}

The above optimization ensures that we find the safe action $\bar a_t$ close to the action $a_t$ generated by the control policy. This ensures that the learning of the control policy $\pi_\theta$ by the reinforcement learning algorithm remains stable.

The function $\mathtt{get\_local\_dynamics}$ in step~\ref{lf1} of Algorithm~\ref{algo1} in Appendix~\ref{sec-KBSE-algorithm} solves the Equation~\eqref{eq:regr}.
The function $\mathtt{get\_safe\_action }$ in line~\ref{lf2} of Algorithm~\ref{algo1} in Appendix~\ref{sec-KBSE-algorithm} solves the Equation~\eqref{eq:qpprob}.

 

\subsection{Properties of KBSE}
\label{sec-lbub}

With more exploration, the estimate of the conditional expectation in condition (iii) of the barrier function improves.
In the following theorem, we provide a bound on the convergence of the approximation of the conditional expectation in the barrier function to the true expectation as the sample size increases.

\begin{thm}
\label{thm-conv}
For a given confidence parameter $\zeta$, the approximation precision $\epsilon$, representing the MMD radius, decreases as the number of samples increase with the upper bound $\sqrt{\frac{C}{N}} (1 + \sqrt{2\cdot \mathtt{log}(\frac{1}{\zeta})})$, where $C$ is a constant such that $\mathtt{sup }(k_{S}(s,s)) \leq C \leq \infty$ and $N$ is the number of samples. This signifies the convergence of empirical CME $\hat{\mu}$ to the true CME $\mu$ with probability at least $(1- \zeta)$. Moreover, for a fixed $\epsilon$, the confidence in the approximation improves with the increase in the sample size $N$ according to $\zeta \leq \mathtt{exp} \left[{-\frac{1}{2}\big(\frac{\epsilon \cdot \sqrt{N}}{\sqrt{C}}-1\big)^2}\right]~.$
\end{thm}
\begin{proof}
   Refer Appendix~\ref{sec-proofs}.
\end{proof}

\begin{table*}[t]
\begin{center}
\resizebox{0.99\textwidth}{!}{
\begin{tabular} {|l|rr|r|rrrrr|rr|r|}
\toprule
   \multirow{2}{*}{Env.} & \multirow{2}{*}{$p$} & \multirow{2}{*}{$q$}  & \multirow{2}{*}{Safety specification} & \multicolumn{5}{c|}{Training Time $(\si{\second})$} & \# Safety & $90^{th}$ & Safety\\
    &  &    &  & DLag & SLag & DPID & SD & KBSE  & Viol. &  $\%ile$ & Prob.\\
  \midrule
   \textsf{SafetyPendulum}  & 3 & 1  & $ (\vartheta>-0.8)$ & $406$ & $704$ & $423$ & $718$ &  $835$ & $4211$ & $81.60$ & $0.643$\\
    \textsf{SafetyMCar}  & 2 & 1  & $ (p>-1.0) $   & $443$ & $682$ & $486$ & $662$ & $683$   & $342$ & $82.14$ & $0.572$\\
    \textsf{SafetyInvPendulum} & 4 & 1 & $(|p|<0.3)$  &  $2970$ & $3395$ & $2924$ & $3362$ & $3592$  &  $24$ & $7.05$ & $0.972$\\
  \textsf{SafetyHopper}  & 11 & 3  & $ ((z>0.8) \land  (|a|<0.2))$    & $9577$ & $10492$ &  $9312$ & $10520$  &  $12054$  & $5250$ & $68.66$ & $0.958$\\
   \textsf{SafetyWalker} & 17 & 6  & $ ((z>0.9)  \land (|a|<1))$  & $9391$  & $10682$  & $9458$  & $10508$  &  $14196$ & $7289$ & $27.33$  & $0.915$\\
  \textsf{SafetyAnt} & 27 & 8  & $ (z>0.25)$   &  $9438$ & $10472$ & $9637$ &  $10608$ & $24208$  & $105113$ & $52.60$  & $0.729$\\
    \textsf{SafetyHumanoid} & 376 & 17 &  $(z>1.05)$   & $15838$ & $17585$ & $15791$ & $17848$ &  $30452$   & $923$ & $20.60$ &  $0.872$\\
 \bottomrule
\end{tabular}
}
\caption{State dimensions ($p$), action dimensions ($q$), Safety specifications, computation time and number of safety violations for different environments (Env). For safety specifications $p$: position of the robot, $z$: height of the center of mass, $a$: pedal angle.}
\label{table-bench}
\end{center}
\vspace{-0.5cm}
\end{table*}

\begin{table}[t]
  \begin{center}
   \resizebox{\columnwidth}{!}{
 \begin{tabular}{| l | l | r | r | r|| l | l | r | r | r|} 
    \toprule 
 Env. & Algo.  & Avg. Reward & Avg. Cost & Avg. Length & Env. & Algo.  & Avg. Reward & Avg. Cost & Avg. Length  \\ 
 \midrule
  \multirow{5}{4em}{\textsf{Safety Pendulum}}  & DDPGLag & $-143.62$ & $6.8$ &  $200$ &  \multirow{5}{4em}{\textsf{Safety MountainCar}}   & DDPGLag & $-0.06$ & $0.0$ & $1000.0$  \\ 
  & SACLag &  $-120.15$  &  $4.2$  &  $200$  &  & SACLag &  $93.38$  & $12.0$   & $137.2$  \\  
   & DDPGPID & $-145.75$   &  $6.4$  & $200$ &  & DDPGID & $-0.01$   & $0.0$   &  $1000.0$  \\ 
    & SACPID &  $-162.89$  &  $8.6$  & $200$ &  & SACPID & $93.96$   & $12.8$   & $144.4$   \\ 
     & KBSE &  $-164.68$  & $7.8$   &  $200$ & & KBSE &  $93.04$  & $18.2$   & $150.4$   \\ 
    \midrule
     \multirow{5}{4em}{\textsf{Safety InvertedPendulum}}   & DDPGLag & $1000.0$ & $6.6$ & $1000.0$ & \multirow{5}{4em}{\textsf{Safety Hopper}}  & DDPGLag & $1204.37$  &  $3.4$ & $404.0$\\ 
  & SACLag &  $1000.0$  &  $0.0$  & $1000.0$  & & SACLag &  $2991.61$  &  $0.0$  & $1000.0$ \\  
   & DDPGID &  $1000.0$  & $6.6$   & $1000.0$  &  & DDPGPID &  $2110.11$  &  $0.8$ & $613.4$  \\ 
    & SACPID &  $1000.0$  & $0.0$   & $1000.0$  &  & SACPID &  $1035.41$  & $0.0$ & $1000.0$  \\ 
     & KBSE &  $1000.0$  & $0.0$   & $1000.0$  &  & KBSE &  $3203.76$  & $0.0$ & $1000.0$ \\ 
 \midrule
\multirow{5}{4em}{\textsf{Safety Walker}}  & DDPGLag & $2042.48$  & $21.2$   &  $612.4$  & \multirow{5}{4em}{\textsf{Safety Ant}}  & DDPGLag & $257.69$  & $164.0$ &  $252.4$ \\  
  & SACLag & $1734.94$ & $9.0$ & $558.60$ &  & SACLag & $-283.28$   &  $411.8$  &   $867.0$\\   
   & DDPGID &  $935.50$  &  $7.6$  &  $322.2$  & & DDPGID & $1403.84$   &  $50.0$  & $748.4$   \\ 
    & SACPID & $450.34$   &  $83.8$  &  $577.2$ &   & SACPID &  $-1257.27$  & $315.8$   & $1000.0$  \\ 
     & KBSE &  $3819.62$  &  $0.0$  &  $1000.0$ &  & KBSE & $5819.89$   & $0.0$   &  $1000.0$  \\ 
     \midrule    
\multirow{5}{4em}{\textsf{Safety Humanoid}}   & DDPGLag & $1362.88$  & $5.6$ & $275.8$ \\ 
  & SACLag & $ 5469.56$   & $0.0$   & $1000.0$   \\  
   & DDPGID & $2217.06$   & $5.2$   & $432.0$   \\ 
    & SACPID &  $5169.28$   & $0.0$   &  $1000.0$  \\ 
     & KBSE &  $5483.47$  &  $0.0$  & $1000.0$   \\ 
 \cline{1-5}
  \end{tabular}
 }
   \caption{Comparing the average reward, cost, and length  of the KBSE w.r.t baseline algorithms for all benchmarks.}
  \label{table-full}
  \end{center}
\vspace*{-1em}
\end{table}

\section{Evaluation}
We now describe the setup and results for our experiments.

\smallskip
\noindent
\textbf{Experimental setup.}
All experiments are carried out on an Ubuntu22.04 machine with Intel(R) Xeon(R) Gold 6226R $@ 2.90$ GHz$\times16$ CPU, NVIDIA RTX A4000 32GB Graphics card, and 48\,GB RAM using the OMNISAFE~\citep{omnisafe} framework.

\smallskip
\noindent
\textbf{Baseline for comparison.}
Since KBSE belongs to the class of off-policy safe RL algorithms, we consider the most popular off-policy safe RL algorithms as baseline for comparison, the Lagrangian and PID-Lagrangian approaches, i.e.,
DDPGLagrangian, SACLagrangian~\citep{Ray2019} and DDPGPID, SACPID ~\citep{stooke20}.
We have not reported the comparison results with other safe RL algorithms such as CPO~\citep{Achiam17}, PPO-Lag~\citep{Ray2019}, or CUP~\citep{Yang22} as these are \emph{on-policy} algorithms and are known to suffer from high variance and slow convergence compared to the off-policy methods~\citep{chung21a}.


\smallskip
\noindent
\textbf{Benchmarks.}
Our experiments include classical and  Gym-Mujoco benchmarks~\citep{gym} involving safety constraints. See Table~\ref{table-bench} with the details of the 
benchmarks described in the supplementary material and the hyper-parameters in Appendix~\ref{sec-kbse-param}.


\subsection{Results}
\textbf{Training.}
The training time needed to learn the control policy and the number of safety violations encountered during training for the KBSE algorithm are presented in Table~\ref{table-bench}.
The column $90^{th} \%ile$ provides the time-step by which $90\%$ of the total number of safety violations have occurred for the KBSE algorithm, as a percentage of the training horizon.
This indicates that the number of safety violations decreases significantly in the later stages of training. Given that the baselines do not permit quantitative safety violations in their problem formulation, the metrics on safety violations (e.g., 90th percentile) are not applicable in these cases.

We also observe that the training time for KBSE algorithm increases with increasing number of safety violations. This is most apparent in case of the SafetyAnt benchmark where the training time is twice than that of the baselines.
The training time depends on the number of safety violations encountered during exploration (cf. Appendix~\ref{sec-cost-viol}). This, in turn, depends on the shape of the barrier function - one with narrow "safe set" (i.e. stricter safety specifications) would cause more violations. 
The last column of Table~\ref{table-bench} provides the lower bound on the safety probability guaranteed by the KBSE algorithm for each benchmark, a feature that is not integrated in the available safe RL algorithms.

\smallskip
\noindent
\textbf{Testing.}
Table~\ref{table-full} compares the control policies learned by the KBSE algorithm w.r.t. the baselines.
For testing policies, we use the metrics \emph{Average Episodic Reward}, \emph{Average Episodic Cost}, and \emph{Average Episodic Length}. 
 In general, the results demonstrate that the KBSE algorithm achieves higher reward and lower cost compared to the baselines.
In benchmarks with high dimension, the KBSE algorithm shows superior performance that can be attributed to the larger exploration space for these high-dimensional benchmarks, providing a significant margin for improvement.
On the other hand, in low-dimensional benchmarks such as SafetyPendulum and Safety-MountainCar, the performance of KBSE is similar to the baselines as the training occurs for a shorter duration and hence the margin of improvement is smaller.

Further experimental results are available in the appendix, showing synthesized barrier functions (Appendix~\ref{sec-viz-barr}), plots for $\epsilon$ and $\bar{B}$ (Appendix~\ref{sec-plot-ep}), the relationship between training time and safety violations (Appendix~\ref{sec-cost-viol}).

\section{Conclusion}

We have proposed a novel safe exploration algorithm using online barrier functions constructed using kernel mean embeddings (CME). 
Our approach does not require knowledge of the system dynamics and uses barrier functions that involve chance constraints that allow the violation of the 
safety specification up to a given probability threshold. Future work includes enhancing the performance using sparse CME and considering temporal behaviors beyond safety.

\section{Acknowledgments}
This work was supported by the European Research Council (ERC, grant 101089047), the European Innovation Council (EIC, grant 101070802), and the Deutsche Forschungsgemeinschaft (DFG, project 389792660, TRR 248—CPEC).

\bibliography{references}

\newpage
\appendix
\newpage
\pagebreak

\appendix

\begin{center}
{\Huge Appendix}
\end{center}

\setcounter{secnumdepth}{2}

\setcounter{equation}{0}

\section{KBSE Algorithm}
\label{sec-KBSE-algorithm}
\subsection{Algorithm}
The overall algorithm, called \emph{kernel-based Safe Exploration} (\textit{KBSE}) and presented in Algorithm~\ref{algo1}, learns a controller online in a Deep-RL setting for the CMDP $\mathcal M$ with an unknown probability kernel $\mathcal T$.

The initialization for the algorithm is performed in lines~\ref{l11}-\ref{lf0}.
The dataset $\mathcal R$,
the initial state $s_0$,
the bound on the norm of the barrier function $\bar B$
and the approximation precision $\epsilon$
are initialized in line~\ref{l11}.
A random policy is instantiated in line~\ref{l12}. Under this policy, a collection of sampled transitions of the form $\langle s_t,a_t,r_t,s_{t+1} \rangle$ are collected in $w$ in line~\ref{l13}, and a valid barrier function together with its associated constants $(\nu,\eta,c)$ are computed in line~\ref{lf0}. The sampled transitions $w$ are added to the dataset $\mathcal R$ in line~\ref{lf_update_data}.

The lines~\ref{l132}-\ref{l25} describe the main loop of the algorithm with the training horizon $T_h$. In line~\ref{l_safe}, we check for a safety violation at state $s_t$.
In case of safety violation, the algorithm invokes
the function $\mathtt{get\_local\_dynamics }$ in line~\ref{lf1} to obtain the local dynamics of the system and use it to
find the closest safe action using the function $\mathtt{get\_safe\_action }$ in line~\ref{lf2}.
In line~\ref{l15}, the system is simulated to generate a transition sample and appends it to the dataset $\mathcal{R}$ (line~\ref{l16}). In line~\ref{l17}, a batch of data samples from $\mathcal{R}$ is used to update the policy, similar to the policy update in offline reinforcement learning~\citep{Levine20}.

The lines~\ref{l19}-\ref{lf4} are executed after each episode to update the approximation error of the condition (iii) of the barrier function. This involves recomputing the upper bound of RKHS norm of the barrier function $\bar{B}$ (line~\ref{lf3}) and the approximation precision $\epsilon$ (line~\ref{lf4}) using the updated dataset (line~\ref{l19}) and confidence parameter $\zeta$.  
The lines~\ref{l21}-\ref{lf5} are executed at the end of each epoch, which is over a time horizon longer than the episode length. An epoch contains a number of episodes. The condition \texttt{epoch ends} is used to compute the barrier function and \texttt{episode ends} is used to update the barrier function
parameters. The function  $\mathtt{compute\_BC} $ in line~\ref{lf5} involves learning the barrier function over the latest dataset from line~\ref{l21} and performing a validity check iteratively until a valid barrier function is obtained. Although the iterative check entails an overhead, but it is much less expensive and more scalable than computing the barrier function a priori or offline. The parameters of the learned barrier function will give the bound on the safety probability $(1-\delta)$ using the line~\ref{line:safety_prob}.

\RestyleAlgo{ruled,vlined,linesnumbered}
\LinesNumbered

\begin{algorithm2e}[!ht]
\caption{{KBSE algorithm for CMDP}}
\label{algo1}

\DontPrintSemicolon 
\LinesNumbered

\SetKwProg{myproc}{procedure}{}{}

\myproc{$\mathtt{main}$()}{

\textbf{Input:} CMDP $\mathcal M = \tuple{S, A, \mathcal{T}, R}$ with unknown $\mathcal T$, safety specification $\varphi = \tuple{S_u, T}$, parameterized policy $\pi_\theta$, training horizon $T_h$, confidence parameter $\zeta\in(0,1)$ for the CME

\label{l10}

$\mathcal{R} \gets []$;\, $s_0 \gets \rho_0$;\,$\bar{B} \gets \mathtt{initialize }()$; \,
$\epsilon \gets\mathtt{initialize }(\zeta)$;\,
\label{l11} 


$\pi_{\theta} \gets \mathtt{initialize\_policy}()$\;\label{l12}

$w \gets \mathtt{initial\_rollout }(\pi_{\theta})$\;\label{l13}

$(B,\nu,\eta,c) \gets \mathtt{compute\_BC }(w,\epsilon,\bar{B})$\;\label{lf0}

$\mathcal{R} \gets \mathcal{R}  \cup w$\;\label{lf_update_data}


\While{$t~<~T_h$}{\label{l132}
$a_t \gets \pi_{\theta}(s_t)$\; \label{l14}

\If{$s_t\in S_u$}{\label{l_safe}
/* Safety is violated */\;
$P, Q \gets \mathtt{get\_local\_dynamics }(\mathcal{R})$\;\label{lf1}

$\bar{a}_t \gets \mathtt{get\_safe\_action }(B,s_t,a_{t},P,Q)$\;\label{lf2}
$a_t \gets \bar{a}_t$\;\label{lf30}
}

$s_{t+1}$ sampled from $\mathcal T(s_t,a_t,\cdot)$, $r_t = R(s_t)$
\label{l15}

$\mathcal{R} \gets \mathcal{R}  \cup \{\langle s_t,a_t,r_t,s_{t+1} \rangle\}$\;\label{l16}


$\mathtt{update\_policy }(\pi_{\theta},\mathcal{R})$\;\label{l17}

\If{$\texttt{episode ends}$}{\label{l18}
$w \gets \mathtt{sample\_data }(\mathcal{R})$\;\label{l19}

$\bar{B} 
\gets \mathtt{update\_upper\_bound }${ ($B,w$)}\;\label{lf3}

$\epsilon 
\gets \mathtt{update\_epsilon }${ ($w,\zeta$)}\;\label{lf4}

}

\If{$\texttt{epoch ends}$}{\label{l20}
$w \gets \mathtt{sample\_data }(\mathcal{R})$\;\label{l21}
$(B,\nu,\eta,c) \gets \mathtt{compute\_BC }${ ($w,\epsilon,\bar{B}$)}\;\label{lf5}
}

$t \gets t + 1$\;\label{l25}
}
$\delta \gets (\eta + cT)/\nu$\; \label{line:safety_prob}
\Return{ Policy ${\pi}_{\theta}$ and barrier function $B$ with safety probability $(1-\delta)$ guaranteed with confidence $(1-\zeta)$}\;  \label{l26}
}

\end{algorithm2e}

\section{Proofs of Theorems}
\label{sec-proofs}

\vskip10pt

\setcounter{thm}{1}
\begin{thm}
    Consider the dataset $ \mathcal{R} = (\hat{S},\hat{A},\hat R, \hat{S}^+)$ with the kernels  $k_S$ and $k_A$ and the Gram matrices $K_{\hat{S}}$, $K_{\hat{A}}$ and $K_+$. 
    Consider the ambiguity set $\mathcal{C}_{\epsilon}$ centered at the empirical CME $\hat\mu$ in \eqref{eq:empirical_CME} with confidence $(1-\zeta)$. If there exists a function $B : S \rightarrow \mathbb{R}_{\geq 0}$, $B \in \mathcal{H}_{k_S}$ with $\bar{B} \geq \| B \|_{\mathcal{H}_{k_S}}$, such that for all $s\in S$ there is an $a\in A$ satisfying
    \begin{equation}
     W(s,a)^T\cdot B(\hat{S}^+) - B(s) \leq c - \epsilon\cdot \sqrt{k_{SA} ((s,a),(s,a))}\cdot \bar{B},
    \end{equation}
    for some $c\geq 0$ with $W(s,a)$ in \eqref{eq:W},
    then B satisfies the condition (iii) of the barrier function with constant $c$ and confidence $(1-\zeta)$.
\end{thm}

\begin{proof}
    The inner product of the barrier function $B$ and CME $\mu(\cdot, s,a)$ in $\mathcal{H}_{k_{S}}$ can be expressed as
    \begin{align*}
        \langle B, \mu(\cdot, s,a) \rangle_{\mathcal{H}_{k_{S}}} \!\!\!=\!\! \langle B, \mu(\cdot,s,a) \!-\! \hat{\mu}(\cdot, s,a) \rangle_{\mathcal{H}_{k_{S}}} \!\!\!+\!\! \langle B, \hat{\mu}(\cdot, s,a) \rangle_{\mathcal{H}_{k_{S}}}.
    \end{align*}
     The second term gives us $W(s,a)^T B(\hat{S}^+)$ from Equation~\eqref{eq-inner-prod}. To bound the first term,
let $\Gamma: (S,A)\times (S,A) \rightarrow \mathcal{L}(\mathcal{H}_{k_{S}})$ be the operator-valued positive definite kernel of $\mathcal{G}$ (cf. Remark~\ref{rem:G} and~\cite{Li22}), given by $\Gamma((s,a),(s',a')) := k_{SA}((s,a),(s',a'))\cdot Id_{\mathcal{H}_{k_{S}}}$, where $\mathcal{L}(\mathcal{H}_{k_{S}})$ is the banach space of bounded linear operators from $\mathcal{H}_{k_{S}}$ to $\mathcal{H}_{k_{S}}$ and $Id_{\mathcal{H}_{k_{S}}}$ is the identity operator on $\mathcal{H}_{k_{S}}$.  
Using the reproducing property of $\Gamma$ (cf.~\cite{Li22}), we get
{\allowdisplaybreaks
    \begin{align*}
         \langle B, \mu - \hat{\mu}\rangle_{\mathcal{H}_{k_{S}}} &= \langle \Gamma(\cdot, (s,a))\cdot B, \mu - \hat{\mu} \rangle_{\mathcal{G}}
         \leq \epsilon \|  \Gamma(\cdot, (s,a))\cdot B \|_{\mathcal{G}} \\
         &= \epsilon \sqrt{\langle \Gamma(\cdot, (s,a))\cdot B, \Gamma(\cdot, (s,a))\cdot B \rangle_{\mathcal{G}}}\\
         &= \epsilon \sqrt{\langle B, \Gamma((s,a), (s,a))\cdot B \rangle_{\mathcal{H}_{k_{S}}}}\\
         &= \epsilon \bar{B} \sqrt{\| \Gamma((s,a), (s,a)) \|}
         = \epsilon \bar{B}  \sqrt{k_{SA} ((s,a),(s,a))}.
    \end{align*}
    }
Combining the two terms concludes the proof.
\end{proof}

\begin{thm}
For a given confidence parameter $\zeta$, the approximation precision $\epsilon$, representing the MMD radius, decreases as the number of samples increase with the upper bound $\sqrt{\frac{C}{N}} \left(1 + \sqrt{2\cdot \mathtt{log}(\frac{1}{\zeta})}\right)$, where $C$ is a constant such that $\mathtt{sup }(k_{S}(s,s)) \leq C \leq \infty$ and $N$ is the number of samples. This signifies the convergence of empirical CME $\hat{\mu}$ to the true CME $\mu$ with probability at least $(1- \zeta)$. Moreover, for a fixed $\epsilon$, the confidence in the approximation improves with the increase in the sample size $N$ according to $\zeta \leq \mathtt{exp} \left[{-\frac{1}{2}\big(\frac{\epsilon \cdot \sqrt{N}}{\sqrt{C}}-1\big)^2}\right]~.$
\end{thm}

\begin{proof}
    From~\cite[Eqn. 4]{Nemmour22}, the error bound of the MMD estimators certify that the true CME $\mu$ is contained in an MMD ball of radius $\epsilon$ around the  empirical CME $\hat{\mu}$ with probability $(1-\zeta)$, where 
    \begin{align*}
        \epsilon &\leq \sqrt{\frac{C}{N}} + \sqrt{\frac{2\cdot C\cdot \mathtt{log}(\frac{1}{\zeta})}{N}}
         = \sqrt{\frac{C}{N}} \left(1 + \sqrt{2\cdot \mathtt{log}(\frac{1}{\zeta})}\right)~.
    \end{align*}
    Rearranging the terms in the above inequality gives the reported bound for $\zeta$.
\end{proof}

\section{Benchmarks}
\label{sec-bench}

\smallskip
\noindent
\textbf{Pendulum.}
For the pendulum example, the safety specification requires that the pendulum never goes for a full swing to achive its goal (cf. Section~\ref{sec-pend}).

\smallskip
\noindent
\textbf{MountainCarContinuous.}
In Mountain Car example, position $p \in [-1.2,0.6]$, the safety constraint requires that the car does not go to the extreme left while gaining momentum to climb the hill, i.e. $p>-1.0$. 
Thus, the safety specification enforces a reduction on the maximum momentum that the car could use to reach the goal. 


\smallskip
\noindent
\textbf{Inverted Pendulum.}
For the inverted pendulum, the safety specification requires that the inverted pendulum never goes too far in the horizontal direction, i.e. outside [-0.3,0.3]. 

\smallskip
\noindent
\textbf{Hopper.}
For Hopper environment, the safety specification requires that the height of its center of mass should always be greater than $0.8$ meters and its pedal angle should always be in range $[-0.2,0.2]$ radians.
We construct a multi-dimensional safety specification defined as $\{z>0.8 \land |a|<0.2\}$.

\smallskip
\noindent
\textbf{Walker.}
For Walker environment, the safety specification requires that the height of its center of mass should always be greater than $0.9$ meters and its pedal angle should always be in range $[-1.0,1.0]$ radians.
We construct a multi-dimensional safety barrier defined as $\{z>0.9 \land |a|<1.0\}$.

\smallskip
\noindent
\textbf{Ant.}
The Ant environment belongs to the quadruped class with the safety specification requiring that the height of its center of mass should always be greater than $0.25$ meters above the ground. This is challenging because the ant can achieve higher speed by bending its legs more (unsafe state), resulting in lower height of center of mass.

\smallskip
\noindent
\textbf{Humanoid.}
The humanoid is the most complex benchmark used in our experiments. The safety specification requires that the height of the center of mass is always above $1.1$ meters above the ground.

\section{Additional Experiments}

\subsection{Visualizing synthesized barriers}
\label{sec-viz-barr}
Figure~\ref{fig:barr} shows an intermediate barrier synthesized during the later stage of exploration via the KBSE algorithm.
The barrier is one-dimensional ($\vartheta$) for the Pendulum and two-dimensional ($z$, $a$) for the Hopper environment. The regions in red and blue correspond to large and small values for the barrier function, respectively. The yellow curve represents the contour for $B(s_t)=\nu$. In case of Pendulum (left), the contour is a line around $\vartheta=-0.8$ which corresponds to the boundary of the safety specification. This indicates that the barrier function successfully captures the safety for the Pendulum environment. Similar argument can be made for the Hopper environment (right).   
%

\begin{figure}[!ht]
    \centering
    \includegraphics[width=0.4\textwidth]{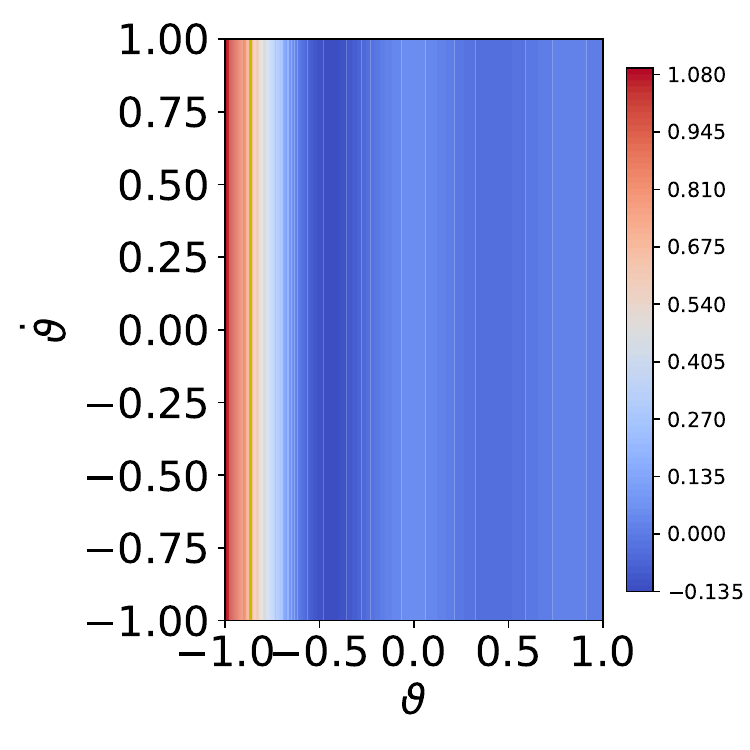}
    \includegraphics[width=0.38\textwidth]{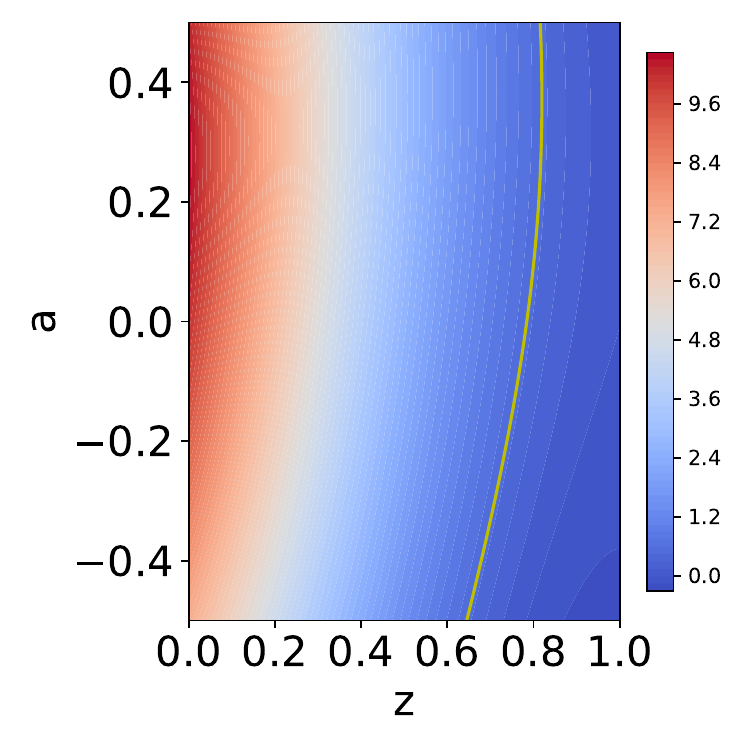}
    \caption{Plots showing an intermediate barrier synthesized for \textsf{SafetyPendulum-v1} (left) and \textsf{SafetyHopper-v1} (right) during the exploration phase for the KBSE algorithm..}
    \label{fig:barr}
\end{figure}

\subsection{Plots for $\epsilon$ and $\bar{B}$}
\label{sec-plot-ep}
Figure~\ref{fig:epub} shows the values of $\epsilon$ and $\bar{B}$ during the exploration via the KBSE algorithm for the Hopper environment. We observe that $\epsilon$ and $\bar{B}$ decrease with time, leading to a reduction in the approximation error introduced due to the  empirical CME.

\begin{figure}[!ht]
    \centering
    \includegraphics[width=0.45\textwidth]{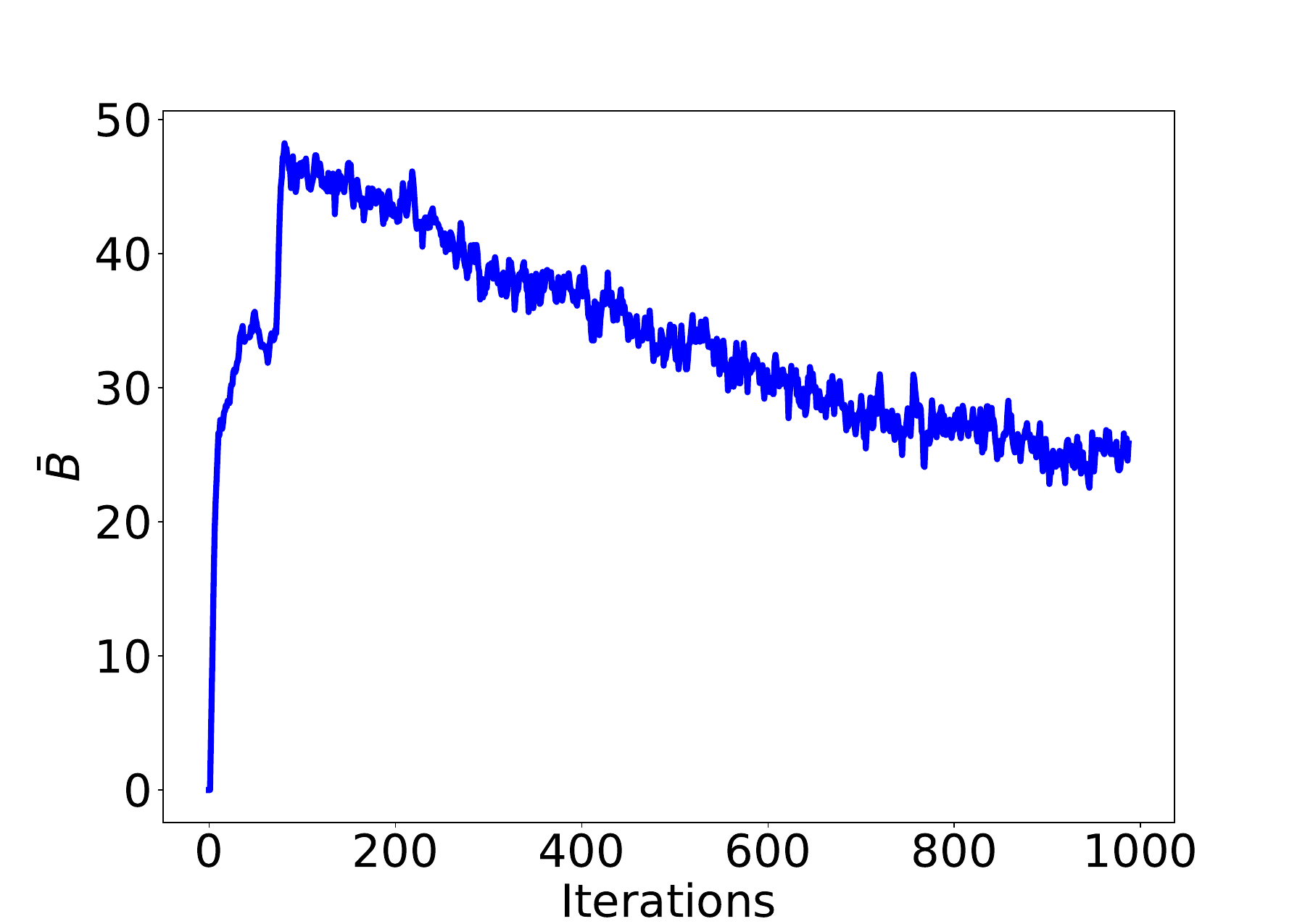}
    \includegraphics[width=0.45\textwidth]{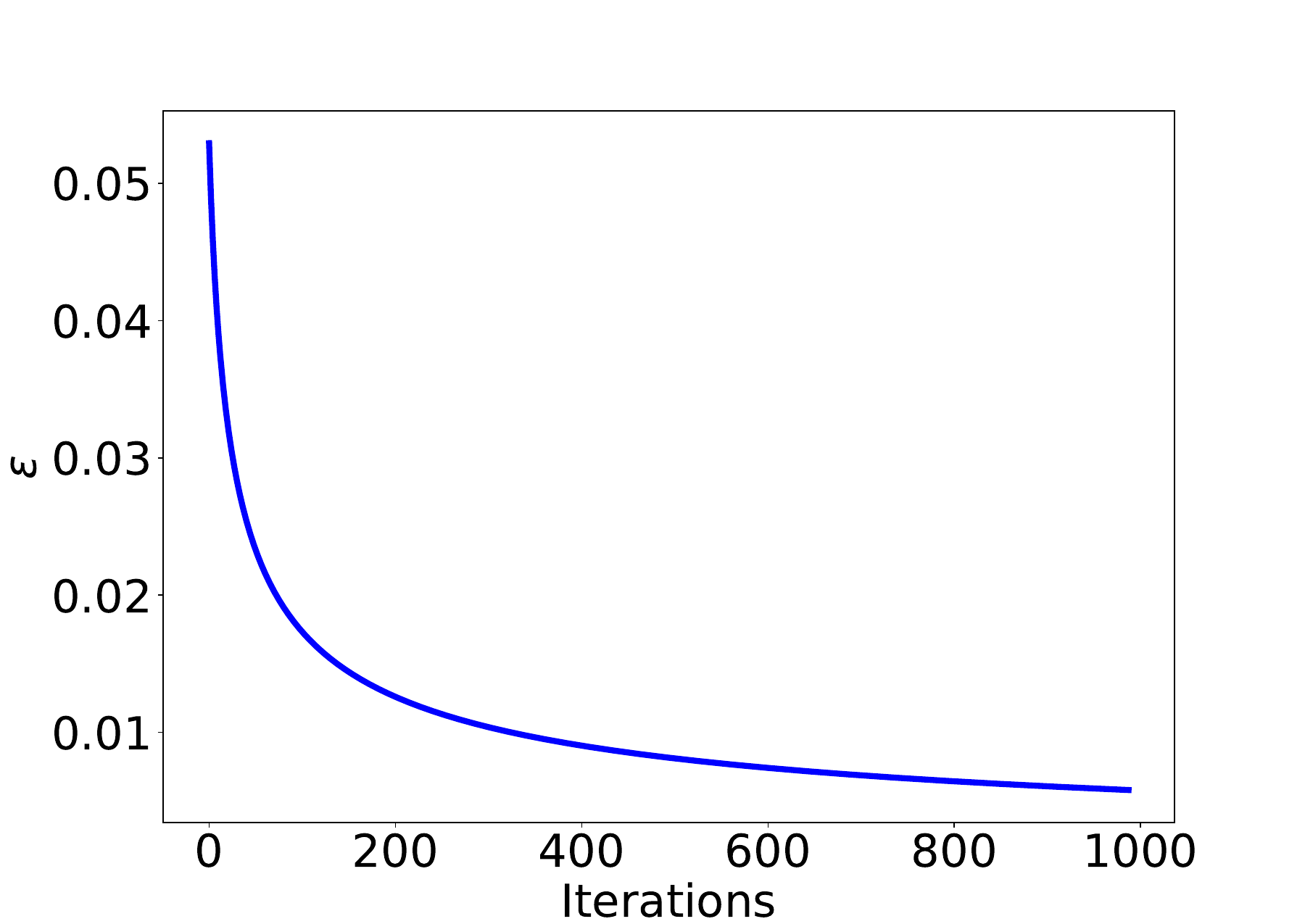}
    \caption{Values of $\bar{B}$ in the left and $\epsilon$ in the right during the exploration phase for the Hopper environment.}
    \label{fig:epub}
\end{figure}

\subsection{Training Time and Safety Violations}
\label{sec-cost-viol}
The training time depends on the number of safety violations encountered during exploration. Hence, relaxing the safety specification reduces training time. To demonstrate this, Table~\ref{table-scost-viol} reports the training time and the number of safety violations for different safety specifications in the \textsf{SafetyHopper} benchmark.

\begin{table}[!ht]
  \begin{center}
  {\fontsize{9}{6}\selectfont
  \begin{tabular}{|l | r | r |}
    \toprule
    Safety Specification  & Training Time  & Safety Violations\\
    \midrule
     $z>0.8 \land |a|<0.2$ &    $20516.26$ &  $3745$        \\
       $z>0.75 \land |a|<0.2$ &   $17901.66$ &     $2155$ \\
       $z>0.70 \land |a|<0.2$ & $14107.36$ & $0$ \\
     \bottomrule
  \end{tabular}
  }
  \end{center} 
  \caption{Comparison of the training time and the number of safety violations for the \textsf{SafetyHopper} benchmark. 
  }
  \label{table-scost-viol}
\end{table}





\section{Hyper-parameters for KBSE Algorithm}
\label{sec-kbse-param}
In our experiments, we have used the same configuration as used by the OMNISAFE framework~\citep{omnisafe}.
We have used $H=500$ in \eqref{eq:regr} to learn the local linear dynamics in order to find the closest safe action. The value of $T$ is equal to the episode length for the corresponding environment. The values for Epoch and Barrier Size ($N$) is shown in Appendix~\ref{sec-kbse-param}. For kernel bandwidth, we followed the same approach as described in~\cite{Duvenaud}. The selected values for the hyper-parameters provide a good balance between accuracy and computational cost.

 The term  $\epsilon\cdot \sqrt{k_{SA} ((s,a),(s,a))}\cdot \bar{B}$ in Equations~\eqref{bar_eq}--\eqref{final-eq} represents the conservatism introduced due to the worst-case analysis with respect to the unknown dynamics $\mathcal T$. 
The function $\mathtt{update\_upper\_bound}$ in Algorithm~\ref{algo1} computes the upper bound of $B$. This is given by the RKHS norm of $B$, i.e., $\| B \|_{\mathcal{H}_{k_S}}=\sqrt{\alpha^T \cdot K_{\hat{S}\hat{A}} \cdot \alpha}$, where $\alpha =  [K_{\hat{S}\hat{A}} +\lambda N \mathbb{I}_N]^{-1}\cdot[B(s_i)]_{i=1}^{N} $ \citep{Romao23}. 
%
%
For the KBSE algorithm, we have used the RBF kernel with $C=1$ and assume $\zeta=10^{-5}$, which means
the function $\mathtt{update\_epsilon}$ in Algorithm~\ref{algo1}, updates the MMD radius $\epsilon$ at the rate of $\sqrt{\frac{1}{N}} \left(1 + \sqrt{10}\right) $. 



For KBSE, the values of the hyper-parameters \texttt{Epoch}, \texttt{Barrier Size} ($N$), and the total number of training time-steps are shown in Table~\ref{table-kbse-param}.

\begin{table}[!ht]
  \begin{center}

  {\fontsize{9}{6}\selectfont
  \begin{tabular}{|l | r | r |r|}
    \toprule
   
    Env  & Epoch &  $N$ & Training steps \\
    \midrule
   \textsf{SafetyPendulum}  & $10000$ & $500$ & $50000$ \\
    \midrule
   \textsf{SafetyMountainCar} & $10000$ & $500$  & $50000$ \\
   \midrule
   \textsf{SafetyInvPendulum} & $10000$ & $2000$ & $200000$\\
   \midrule
   \textsf{SafetyHopper} & $10000$ & $2000$ & $600000$\\
   \midrule
   \textsf{SafetyWalker2d} & $10000$ & $2000$ & $600000$\\
   \midrule
   \textsf{SafetyAnt} & $10000$ & $2000$ & $600000$\\
   \midrule
   \textsf{SafetyHumanoid} & $50000$ & $2000$ & $1000000$\\
     \bottomrule
  \end{tabular}
  }
  \caption{Hyper-parameters used in the KBSE algorithm.}
  \label{table-kbse-param}
  \end{center}
\end{table}

\section{Additional Discussions}

\noindent
\textbf{Limitations.}
A general limitation of kernel-based methods is scalability and high memory requirements. 
Although the computational requirements of our algorithm increase as the dimension of state and action space increase, the increase in memory requirement is negligible. 
This is because the number of data samples used to construct the barrier function online is fixed (the sample set size grows with more data, thereby improving the barrier estimate). The computational overhead for KBSE comes from the quadratic optimization problem for $\texttt{get\_safe\_action}$. While our approach can handle benchmarks similar to \textsf{SafetyHumanoid} that has hundreds of states (cf. Table~\ref{table-bench}), scalability can be further improved by employing sparse CME (e.g., using the work by~\cite{Cortes17}), which we plan to work on in the future. 

\smallskip
\noindent
\textbf{Safety-performance tradeoff.}
Under the proposed KBSE algorithm, exploration can be guided away from highly rewarding but unsafe areas of the state space towards safer regions. The confidence level parameter $\zeta$ in the barrier formulation can be utilized to balance the trade-off between safety and performance.

\smallskip
\noindent
\textbf{Implications of Assumption in Remark~\ref{rem:G}.}
This is an assumption required for the theory of CME to hold and can be verified by having additional information on the class of system dynamics and the distribution of the noise. Under this additional information, one would need to check that the CME of the system dynamics belongs to the chosen Hilbert space. This restricts the conditional distribution to be ``regular enough'' so that it does not concentrate mass in a pathological way or have heavy tails. We refer to \cite{Park20} for the technical details and remark that the requirements such as bounded kernel moments and RKHS norm integrability hold when selecting the kernel to be Gaussian or Radial Basis Functions (RBFs), as employed in our experiments.

\smallskip
\noindent
\textbf{Safety violation as indicator of safe exploration.}
Our problem formulation allows violation of safety up to a certain probability threshold even for the policy learned at the end of the exploration phase. We use safe exploration in the sense of satisfying the constraints characterizing the barrier function. The Safety violations during training occur due to the constraint optimization problem being infeasible. However, the training under the KBSE algorithm ensures that the collected data samples satisfy the constraints, and thus the intermediate policies learned using these data samples would gradually satisfy these constraints.

\smallskip
\noindent
\textbf{Online vs. Off-policy.}
In this paper, the term ``online'' refers to the computation of the barrier function which is synthesized on-the-fly, i.e., during exploration in RL. The ``off-policy'' term refers to the RL based learning paradigm where policy is updated using the older (off-policy) data collected during exploration since the beginning.



\smallskip
\noindent
\textbf{Bias-Variance decomposition of error rates.}
As per Theorem~\ref{thm-conv}, the bias decreases as the number of samples increases. The i.i.d. generation of data ensures that variance is bounded. A large value of the RBF kernel bandwidth makes the model less complex which may reduce variance but increase bias. 

\smallskip
\noindent
\textbf{Assumption on bounds of the RKHS norm of B.}
The assumption of a bounded RKHS norm for the barrier function is a standard regularity condition commonly used in kernel-based learning methods to derive finite-sample generalization bounds ~\citep{Li22,Hou23}, and should be interpreted as such rather than as a claim that an arbitrarily complex safety boundary can always be captured.
If the chosen kernel is insufficiently expressive, the theoretical safety guarantee holds with respect to the best approximation of the true dynamics and barrier within the RKHS, reflecting the standard trade-off between model expressiveness and statistical guarantees.

\smallskip
\noindent
\textbf{Solution to Problem~\ref{prob-def}.}
 KBSE does not solve Problem~\ref{prob-def} in the sense of finding the optimal policy. Instead, it produces a policy $\pi_{\theta}$ accompanied by a barrier $B$ such that Theorem~\ref{thm:safe-prob} certifies $\mathbb{P}
    \{\mbox{some state is in } S_u\}\leq \delta $ with confidence $1-\zeta$, where $\delta=\frac{\eta + cT}{\nu}$. 
    The optimality of   $\pi_{\theta}$ with respect to the reward objective under this safety constraint is not guaranteed; nevertheless, the learned policy is probabilistically safe with confidence $1-\zeta$ and accumulates competitive empirical reward across all benchmarks (Table~\ref{table-full}). 


\smallskip
\noindent
\textbf{Motivation for selecting the RBF kernel.}
We choose the RBF kernel for three reasons. First, it is a universal kernel on compact domains, ensuring that any continuous barrier function can be approximated within $\mathcal{H}_{k_{S}}$ to arbitrary precision. Second, it is a characteristic kernel, which guarantees that the MMD is a proper metric on distributions and that the ambiguity set $\mathcal{C}_{\epsilon}$ in Eq.~\eqref{eq:amb} correctly captures uncertainty about the true CME. Third, it satisfies $\sup_s k_{S}(s,s) =1$, which gives the tightest possible constant $C$ in the convergence bound of Theorem~\ref{thm-conv}, directly controlling the quality of the safety certificate.
\end{document}